\documentclass[11pt,a4paper]{article}
\usepackage{amssymb,amscd,amsmath,amsthm}
\usepackage[colorinlistoftodos]{todonotes}
\usepackage[colorlinks=true, allcolors=blue]{hyperref}
\usepackage{lmodern}
\usepackage[english]{babel}
\usepackage{color}
\usepackage{graphicx}
\usepackage{tikz}
 \usetikzlibrary{shapes,arrows,fit}
\usepackage{xcolor}
 \usepackage{algpseudocode}
 \usepackage{algorithmicx}
 \usepackage{algorithm}
 \usepackage[lmargin=0.7in,rmargin=0.7in,tmargin=0.65in,bmargin=0.75in,footskip=0.25in]{geometry}

\usetikzlibrary{positioning, calc, arrows.meta}
\newtheorem{theorem}{Theorem}[section]

\newtheorem{lemma}[theorem]{Lemma}
\newtheorem{definition}[theorem]{Definition}

\def\Tr{\mathrm{Tr}}

\title{AKLT State is Indeed the Observation Process \\ of  a causal Hidden quantum Markov Model}

\begin{document}

\maketitle
 \date{}

\centerline{ \author{\Large Abdessatar Souissi}}

\centerline{Department of Management Information Systems, College of Business and Economics, }
\centerline{Qassim University, Buraydah 51452, Saudi Arabia}
\centerline{\textit{a.souaissi@qu.edu.sa}}
\vskip0.3cm

\centerline{\author{\Large Amenallah Andolsi}}
\centerline{Nuclear Physics and High Energy Research Unit, Faculty of Sciences of Tunis,}
\centerline{ Tunis El Manar University, Tunis, Tunisia }
\centerline{\textit{amenallah.andolsi@fst.utm.tn }}

\begin{abstract}
We present a rigorous formulation of the spin-1 Affleck–Kennedy–Lieb–Tasaki (AKLT) state within the framework of hidden quantum Markov models (HQMMs). We demonstrate that the AKLT ground state admits a natural representation as the observable output of a causal HQMM, thereby endowing it with an underlying hidden quantum memory consistent with its standard finitely correlated (matrix product state) description. This perspective provides a compact and structurally transparent characterization of the AKLT chain as a quantum spin system equipped with intrinsic quantum memory. Our results suggest that the HQMM perspective may offer a useful framework for investigating measurement-based quantum computation(MBQC).
\end{abstract}


\section{Introduction}
 The AKLT model provides a canonical and analytically tractable example of a one-dimensional quantum spin system in a symmetry-protected topological (SPT) phase~\cite{Chen2011b,Pollmann2010,Pollmann2012,O21,Ragone24}. It is an exactly solvable, translation-invariant spin-1 antiferromagnetic chain (with local Hilbert space $\mathcal{H}_1 \equiv \mathbb{C}^3$) in the Haldane phase~\cite{Haldine83}, whose unique ground state is gapped, exhibits exponentially decaying correlations, and supports fractionalized spin-$\tfrac{1}{2}$ edge modes associated with a virtual Hilbert space $\mathcal{H}_{1/2} \equiv \mathbb{C}^2$~\cite{AKLT1987,Affleck1988}.Beyond its structural role in quantum many-body theory, the AKLT state has gained renewed significance as a resource for measurement-based quantum computation (MBQC), where its structured entanglement and hidden nonlocal (string) order enable universal quantum computation and quantum communication protocols~\cite{SET_MBQC_2025,Gross2007MPS,Wei2011,dar2010,Liu2014,Chen24}. In particular, this hidden order supports the processing of quantum information through adaptive measurement sequences that induce effective dynamical transformations on an underlying virtual space~\cite{Raussendorf2023,MBQC-SPT26}.
 
 Furthermore, the ground state of the AKLT model admits a natural description both as a finitely correlated state (FCS) and as a matrix product state (MPS) ~\cite{FNW92,CPSV21,SA26}.  Concretely, for a finite chain of $n$ spins, the AKLT ground state can be written as the MPS
\[
|\psi_{\mathrm{AKLT}}^{(n)}\rangle
= \sum_{k_1,\dots,k_n \in \{+,0,-\}}
\operatorname{Tr}(A_{k_1}\cdots A_{k_n}) \, |k_1\cdots k_n\rangle,
\]
where the matrices $A_k$ implement the local virtual-to-physical mapping. The expectation value of a local observable $Y$ is given by
\(
\omega_n(Y) = \langle \psi_{\mathrm{AKLT}}^{(n)} | Y | \psi_{\mathrm{AKLT}}^{(n)} \rangle.
\)
In the thermodynamic limit, the finitely correlated state construction~\cite{FNW92} yields
\[
\omega(Y) := \lim_{\substack{m\to\infty \\ p\to\infty}}
\omega_{n+m+p}\bigl(\mathbb{I}^{\otimes m} \otimes Y \otimes \mathbb{I}^{\otimes p}\bigr),
\]
which defines a translation-invariant state on the quasi-local algebra $\mathcal{B}(\mathcal{H}_1)^{\otimes\mathbb{N}}$. Its correlation functions are fully determined by the spectral properties of the associated transfer operator, making the AKLT model a natural framework for operator-algebraic approaches to topological quantum phases and hidden quantum Markov structures.

HQMMs have attracted significant attention over the past decade, motivated by their ability to capture non-classical correlations and memory effects inherent to quantum systems \cite{AGLS24Q, Monras11,Srin2017}. In particular, recent advances in the implementation, simulation, and training of HQMM have established them as promising architectures for quantum machine learning, capable of encoding complex temporal correlations that remain inaccessible to classical stochastic models \cite{SS23,MRDFS22,Li2024}. Beyond their algorithmic relevance, HQMM provide a rigorous operational framework for investigating the structure of quantum memory, its accessibility through measurement, and the fundamental constraints governing its expressive power \cite{FH21, Fanizza24}, revealing deep connections between memory and fundamental physical resources such as compression efficiency, thermodynamic cost, and the emergence of adaptive behavior in quantum agents \cite{Elliot21}. Building on these developments, HQMM-based paradigms are increasingly integrated into broader learning frameworks, including reinforcement learning for quantum processes with hidden memory \cite{TaElM23}. 
A natural hidden Markov model structure on the AKLT chain is captured by an HQMM, which is a quantum stochastic process on \(( \mathcal{B}(\mathcal{H}_{\frac{1}{2}}) \otimes \mathcal{B}(\mathcal{H}_1))^{\otimes\mathbb{N}}\), generated by a completely positive unital map
\[
\mathcal{T} : \mathcal{B}(\mathcal{H}_{\frac{1}{2}}) \otimes \mathcal{B}(\mathcal{H}_{\frac{1}{2}}) \otimes \mathcal{B}(\mathcal{H}_1) \longrightarrow \mathcal{B}(\mathcal{H}_{\frac{1}{2}}),
\]
where the underlying process is a quantum Markov chain that generates the hidden memory of the system~\cite{ASS20,FR15}. From this general three-tensor map, two particular causal orderings emerge as special cases: the conventional ordering \(\mathcal{T}^{(\rm conv)}(a\otimes b\otimes x) = \mathcal{E}_H\bigl(\mathcal{E}_{H,O}(a \otimes b) \otimes x\bigr)\) and the causal ordering \(\mathcal{T}^{(\rm caus)}(a\otimes b\otimes x) =  \mathcal{E}_{H,O}\bigl(\mathcal{E}_H(a \otimes x) \otimes b\bigr)\) for a hidden transition $\mathcal{E}_{H}:  \mathcal{B}(\mathcal{H}_{\frac{1}{2}}) \otimes \mathcal{B}(\mathcal{H}_{\frac{1}{2}})  \rightarrow \mathcal{B}(\mathcal{H}_{\frac{1}{2}}) $ and an emission expectation $\mathcal{E}_{H,O}:    \mathcal{B}(\mathcal{H}_{\frac{1}{2}}) \otimes \mathcal{B}(\mathcal{H}_1) \rightarrow \mathcal{B}(\mathcal{H}_{\frac{1}{2}})$, which yield distinct information flows and correlation structures~\cite{SouBarhqmm2026}. When restricted to commutative subalgebras, HQMMs reduce to classical hidden Markov models~\cite{rab89}, forming a unified framework for classical and quantum hidden memory processes.

Our recent work \cite{SA26} has successfully reformulated the AKLT state within a conventional hidden quantum Markov model (HQMM) framework, thereby translating its symmetry-protected topological (SPT) order and entanglement structure into the language of quantum probability. While this yields a natural HQMM description of the AKLT chain, it simultaneously demonstrates that the AKLT state nevertheless resists realization as an observation process—a result that aligns with the findings obtained for the entangled hidden Markov model (EHMM) \cite{Sou25} within the original entangled-hidden-Markov formalism \cite{SS23}. Motivated by this obstruction, subsequent constructions treat the physical spin‑1 chain as an observation layer of a hidden spin‑\( \tfrac{1}{2} \) system driven by quantum channels closely related to the AKLT transfer matrix, with SPT order encoded as a covariance property of the emission maps under projective symmetry actions. 
In parallel, a causal HQMM   formalism has been developed in \cite{SouBarhqmm2026}, in which the causal order between hidden transitions and emission maps is reversed, leading to a genuinely different class of HQMMs compared with the previously introduced “conventional’’ models \cite{AGLS24Q}. This enlarged landscape of HQMM architectures opens the possibility of associating hidden quantum memory representations to a broader class of many-body states. This brings us to the central question of the present work:
\medskip

\emph{Can the canonical finitely correlated/MPS representation of the AKLT ground state be realised as the observation process of a suitably defined HQMM?}

We begin by constructing the AKLT ground state as an infinite-volume  FCS  on the observable algebra  \cite{FNW92}. This FCS representation naturally allows us to express the model rigorously in the Heisenberg picture, yielding a translation-invariant state $\omega$ fully characterized by an explicit transfer-operator formula for arbitrary local observables. Building on this algebraic foundation, we introduce a class of causal HQMM and establish that their joint hidden–observable states admit a Kraus-type decomposition, thereby making explicit how hidden quantum memory and emission mechanisms cooperate to generate correlations along the chain.

As a consequence, we prove that the AKLT ground state $\omega$ arises precisely as the observation process $\psi_O = \psi_{H,O}\mid_{\mathcal{B}(\mathcal{H}_{1})^{\otimes \mathbb{N}}}$ of a causal HQMM. In sharp contrast, such a realization is provably impossible within the conventional HQMM framework applied to the AKLT state \cite{SA26}, and similarly fails within the EHMMs setting \cite{Sou25}.  This striking separation between causal and conventional HQMM representations reveals a deeper structural distinguishability between the two causal architectures—one that reflects the interplay between quantum hidden memory, measurement ordering, and topological quantum phases—and will be the subject of future investigation.

For ease of reading, the paper is organized as follows: In section~\ref{sect_AKLTFCS}, we construct its the AKLT
ground state as an infinite-volume finitely correlated state/matrix product state (FCS/MPS) on the  observation subsystem. Then, our main result is shown in Section~\ref{sec:HQMM_AKLT}: the AKLT ground state is indeed the observation process of a causal HQMM in an exact operator-algebraic sense. Finally, in Section~\ref{sect-Disc}, we discuss the results and outline future perspectives. The paper concludes with two appendices providing additional details for clarity: Appendix~\ref{app:B} gathers deferred proofs, while Appendix~\ref{app:A} collects background material on HQMMs used throughout the paper. 
 
\section{AKLT finitely Correlated state}\label{sect_AKLTFCS}

The AKLT model provides the paradigmatic example of a gapped, translation-invariant spin-1 chain realising Haldane’s conjecture for integer-spin antiferromagnets \cite{Affleck1988,AKLT1987}. Its Hamiltonian on a one-dimensional lattice with nearest-neighbour interactions is
\[
\hat{H}_{\mathrm{AKLT}}
=
\sum_{\langle i,j\rangle}
\Bigl(
\vec{S}_i\cdot\vec{S}_j
+
\tfrac{1}{3}(\vec{S}_i\cdot\vec{S}_j)^2
\Bigr),
\]
where each $\vec{S}_i = (S_i^x,S_i^y,S_i^z)$ generates the spin-1 representation of $\mathfrak{su}(2)$ via
\[
[S^x,S^y] = i S^z,
\qquad
[S^y,S^z] = i S^x,
\qquad
[S^z,S^x] = i S^y.
\]
The exactly known ground state exhibits short-range correlations, a non-zero spectral gap, and edge modes under open boundary conditions, and admits a particularly transparent representation as a valence-bond solid (VBS) and as a FCS.

The AKLT ground state on a chain of length $n$ can be expressed as a translation-invariant matrix product state with physical local space $\mathcal{H}_{1}\cong\mathbb{C}^3$ spanned by $\{|+\rangle,|0\rangle,|-\rangle\}$ and virtual (bond) space $\mathcal{H}\cong\mathbb{C}^2$ spanned by $\{|\uparrow\rangle,|\downarrow\rangle\}$. The physical–virtual correspondence is encoded in a family of matrices $A_{+},A_{0},A_{-}\in\mathbb{M}_2(\mathbb{C})$ acting on $\mathcal{H}$, given by
\begin{equation}\label{eq:Ak}
A_{+} = \sqrt{\tfrac{2}{3}}\;\sigma^{+},
\qquad
A_{0} = \sqrt{\tfrac{1}{3}}\;\sigma^{z},
\qquad
A_{-} = -\sqrt{\tfrac{2}{3}}\;\sigma^{-},
\end{equation}
where
\[
\sigma^{+} = |\uparrow\rangle\langle\downarrow|,
\quad
\sigma^{-} = |\downarrow\rangle\langle\uparrow|,
\quad
\sigma^{z} = |\uparrow\rangle\langle\uparrow|-|\downarrow\rangle\langle\downarrow|
\]
are the usual Pauli operators acting on $\mathcal{H}$ in the basis $\{|\uparrow\rangle,|\downarrow\rangle\}$. For periodic boundary conditions, the $n$‑site AKLT wave function takes the matrix-product form
\begin{equation}\label{eq:AKLT_MPS_finite}
|\psi_{\mathrm{AKLT}}^{(n)}\rangle
=
\sum_{k_1,\dots,k_n\in\{+,0,-\}}
\Tr\!\bigl(A_{k_1}\cdots A_{k_n}\bigr)\,
|k_1\cdots k_n\rangle,
\end{equation}
with the trace implementing the contraction of virtual indices around the ring and enforcing translation invariance. This construction realises the VBS picture of the AKLT ground state and fits into the general FCS framework, where the bond dimension coincides with the dimension of the auxiliary algebra that carries all correlations across each cut \cite{FNW92}.

The matrices $\{A_k\}$ define a completely positive, unital map (the AKLT transfer channel) acting on the virtual algebra $\mathcal{B}(\mathcal{H}_{\frac{1}{2}})\cong\mathbb{M}_2(\mathbb{C})$ by
We recall that the AKLT transfer channel is defined on $\mathbb{M}_2(\mathbb{C})$ by
\begin{equation}\label{eq:aklt_channel}
\Phi(Z)
=
\sum_{k\in\{+,0,-\}} A_k Z A_k^{\dagger},
\qquad
Z\in\mathbb{M}_2(\mathbb{C}),
\end{equation}
with $A_k$ as in \eqref{eq:Ak}. The AKLT gauge relations
\[
\sum_{k} A_k A_k^{\dagger} = \mathbb{I}_2,
\qquad
\sum_{k} A_k^{\dagger} A_k = \mathbb{I}_2,
\]
Then $\Phi$ coincides with its Hilbert-Schmidt dual i.e.
$$
\Phi(Z)
=
\sum_{k\in\{+,0,-\}} A_k^{\dagger} Z A_k \qquad  Z\in\mathbb{M}_2(\mathbb{C})
$$
means that $\Phi$ is bistochastic: it is unital, $\Phi(\mathbb{I}_2)=\mathbb{I}_2$, and trace-preserving, $\operatorname{Tr}(\Phi(Z))=\operatorname{Tr}(Z)$ for all $Z\in\mathbb{M}_2(\mathbb{C})$.

\begin{lemma}\label{lem:aklt_conv}
Let $\Phi$ be the AKLT  transfer channel defined in \eqref{eq:aklt_channel}. Then
\[
\Phi(\mathbb{I}_2)=\mathbb{I}_2,
\qquad
\Phi(\sigma_\alpha)=-\tfrac{1}{3}\,\sigma_\alpha,\quad \alpha\in\{x,y,z\},
\]
and, for every $M\in\mathbb{M}_2(\mathbb{C})$,
\begin{equation}\label{eq:Phi_power_limit}
\lim_{n\to\infty}\Phi^{n}(M)=\tfrac{\operatorname{Tr}(M)}{2}\,\mathbb{I}_2 = : \Phi_{\infty}(M),
\end{equation}
with exponential convergence at rate $1/3$.
\end{lemma}

We recall the notion of the trace of a linear map (superoperator)
\(F:\mathcal{B}(\mathcal{H}_{\frac{1}{2}})\to\mathcal{B}(\mathcal{H}_{\frac{1}{2}})\) on a finite-dimensional Hilbert space
\(\mathcal{H}\) with orthonormal basis \(\{|\alpha\rangle\}_{\alpha=1}^{d}\). We define
\begin{equation}\label{eq:superop_trace_def}
\operatorname{Tr}_{\mathrm{so}}(F)
:=
\sum_{\alpha,\beta=1}^{d}
\langle \alpha|F(|\alpha\rangle\langle\beta|)|\beta\rangle.
\end{equation}
This quantity is basis-independent and coincides with the ordinary trace of \(F\) viewed as a linear operator on the Hilbert–Schmidt space \(\mathcal{B}(\mathcal{H}_{\frac{1}{2}})\) equipped with the inner product \(\langle X,Y\rangle_{\mathrm{HS}}=\operatorname{Tr}(X^{\dagger}Y)\).\footnote{For any orthonormal basis \(\{E_{ij}=|i\rangle\langle j|\}_{i,j}\) of \(\mathcal{B}(\mathcal{H}_{\frac{1}{2}})\) in the Hilbert–Schmidt inner product one has \(\operatorname{Tr}_{\mathrm{so}}(F)=\sum_{i,j}\langle E_{ij},F(E_{ij})\rangle_{\mathrm{HS}}\).}

\begin{lemma}\label{lem:hatY_properties}
Let \(n\in\mathbb{N}\), let \(\{A_k\}_{k\in\{+,0,-\}}\subset\mathcal{B}(\mathcal{H}_{\frac{1}{2}})\) be the AKLT matrices, and let \(|\psi_{\mathrm{AKLT}}^{(n)}\rangle\in\mathcal{H}_{1}^{\otimes n}\) be the length-\(n\) AKLT MPS defined in \eqref{eq:AKLT_MPS_finite}. For each observable \(Y\in\mathcal{B}(\mathcal{H}_{1})^{\otimes n}\) we define a linear map \(\widehat{Y}:\mathcal{B}(\mathcal{H}_{\frac{1}{2}})\to\mathcal{B}(\mathcal{H}_{\frac{1}{2}})\) by
\begin{equation}\label{eq:hatY_def_lemma}
\widehat{Y}(M)
:=
\sum_{\substack{k_1,\dots,k_n\\ \ell_1,\dots,\ell_n}}
\langle k_1,\dots,k_n|Y|\ell_1,\dots,\ell_n\rangle\,
A_{k_n}^{\dagger}\cdots A_{k_1}^{\dagger}\,M\,A_{\ell_1}\cdots A_{\ell_n},
\qquad M\in\mathcal{B}(\mathcal{H}_{\frac{1}{2}}).
\end{equation}
Then:
\begin{enumerate}
\item[(a)] For every \(Y\in\mathcal{B}(\mathcal{H}_{1})^{\otimes n}\), the length-\(n\) AKLT expectation value is equal to the superoperator trace of \(\widehat{Y}\),
\begin{equation}\label{eq:AKLT_expectation_hatY_correct}
\langle\psi_{\mathrm{AKLT}}^{(n)}|Y|\psi_{\mathrm{AKLT}}^{(n)}\rangle
=
\operatorname{Tr}_{\mathrm{so}}(\widehat{Y})
\end{equation}
  In particular, for the identity observable \(Y=\mathbb{I}^{\otimes n}\) one has \(\widehat{\mathbb{I}^{\otimes n}}=\Phi^{n}\), where \(\Phi:\mathcal{B}(\mathcal{H}_{\frac{1}{2}})\to\mathcal{B}(\mathcal{H}_{\frac{1}{2}})\) is the AKLT transfer channel
\[
\Phi(Z)=\sum_{k\in\{+,0,-\}}A_k Z A_k^{\dagger},\quad Z\in\mathcal{B}(\mathcal{H}_{\frac{1}{2}}),
\]
and therefore
\[
\langle\psi_{\mathrm{AKLT}}^{(n)}|\psi_{\mathrm{AKLT}}^{(n)}\rangle
=
\operatorname{Tr}_{\mathrm{so}}(\Phi^{n}).
\]

\item[(b)] If \(Y\in\mathcal{B}(\mathcal{H}_{1})^{\otimes n}\) and \(Z\in\mathcal{B}(\mathcal{H}_{1})^{\otimes m}\) act on disjoint consecutive blocks, then
\begin{equation}\label{eq:hatY_composition}
\widehat{Y\otimes Z}
=
\widehat{Y}\circ\widehat{Z}.
\end{equation}
\end{enumerate}
\end{lemma}

  For each $n\in\mathbb{N}$, the $n$‑site AKLT state induces a positive linear functional
\[
\omega_n:\mathcal{B}(\mathcal{H}_{1})^{\otimes n}\to\mathbb{C},
\qquad
\omega_n(Y)
=
 \Bigl\langle\psi_{\mathrm{AKLT}}^{(n)}|Y|\psi_{\mathrm{AKLT}}^{(n)}\Bigr\rangle = \operatorname{Tr}_{\mathrm{so}}(\widehat{Y})
\]
which, for elementary tensors $Y = Y_1\otimes\cdots\otimes Y_n$ with $Y_j\in\mathcal{B}(\mathcal{H}_{1})$, admits the explicit expansion
\begin{eqnarray*}
\langle\psi_{\mathrm{AKLT}}^{(n)}|Y|\psi_{\mathrm{AKLT}}^{(n)}\rangle &=
\sum_{\substack{k_1,\dots,k_n\\ \ell_1,\dots,\ell_n}}
(Y_1)_{k_1\ell_1}\cdots(Y_n)_{k_n\ell_n}\,
\Tr\!\bigl(A_{k_1}\cdots A_{k_n}\bigr)\,
\overline{\Tr\!\bigl(A_{\ell_1}\cdots A_{\ell_n}\bigr)}
\end{eqnarray*}
where $(Y_m)_{k_m\ell_m}=\langle k_m|Y_m|\ell_m\rangle$ in the basis $\{|+\rangle,|0\rangle,|-\rangle\}$. This formula makes manifest the virtual–physical duality: physical correlators arise from correlated paths in the auxiliary space generated by the matrices $\{A_k\}$ and their adjoints, with complex conjugation reflecting the bra–ket structure.

The family $(\omega_n)_{n\in\mathbb{N}}$ is not projectively consistent in the sense that, in general, $\omega_{n+1}(Y\otimes\mathbb{I})\neq\omega_n(Y)$ for local observables $Y$, a feature that encodes boundary effects and the non-product nature of the VBS state in finite volume. To obtain a well-defined state on the infinite chain, one must construct a thermodynamic limit by embedding a finite observable into the bulk of an ever-longer chain and taking the limit of the corresponding expectations. Concretely, for $Y\in\mathcal{B}(\mathcal{H}_{1})^{\otimes n}$ we define
\begin{equation}\label{omega_AKLT_limit}
\omega(Y)
:=
\lim_{\substack{m\to\infty\\ p\to\infty}} \omega_{n+m+p}(\mathbb{I}^{\otimes m}\otimes Y\otimes\mathbb{I}^{\otimes p})
\end{equation}
whenever the limit exists. The next result shows that this limit defines a translation-invariant pure state which can be written in closed form.

\begin{theorem}\label{thm:AKLT_infinite_volume}
Let \(Y\in\mathcal{B}(\mathcal{H}_{1})^{\otimes n}\) be an observable supported on \(n\) consecutive sites. Then the limit in \eqref{omega_AKLT_limit} exists and is independent of the way \(m,p\to\infty\) are taken, and the resulting functional \(\omega\) extends uniquely to a translation-invariant,  state on the quasi-local C\(^*\)-algebra \(\mathcal{B}(\mathcal{H}_{1})^{\otimes\mathbb{N}}\). Moreover, for every such \(Y\) we have the explicit formula
\begin{equation}\label{eq:omega_infinite_formula_second_linear}
\omega(Y)
=
\frac{1}{2}\sum_{\substack{k_1,\dots,k_n\\ \ell_1,\dots,\ell_n}}
\langle k_1,\dots,k_n|Y|\ell_1,\dots,\ell_n\rangle\,
\Tr\!\bigl(
 A_{k_1}\cdots A_{k_n}A_{\ell_n}^{\dagger}\cdots A_{\ell_1}^{\dagger}
\bigr)
\end{equation}
where indexes of the sums run in the local  basis  \(\{|+\rangle,|0\rangle,|-\rangle\}\)  of \(\mathcal{H}_{1}\), the matrices \(\{A_k\}\) are those in \eqref{eq:Ak}.
\end{theorem}
\begin{proof}
Consider the finite chain of length \(N=m+n+p\) with periodic boundary conditions, and let \(Y\) act on the block of sites \(\{m+1,\dots,m+n\}\). Using the MPS representation of the AKLT ground state \(|\psi_{\mathrm{AKLT}}^{(N)}\rangle\) given in \eqref{eq:AKLT_MPS_finite} and contracting all physical indices outside the support of \(Y\), one obtains a transfer-operator expression for the finite-volume expectation. Lemma~\ref{lem:hatY_properties} yields the identity
\begin{equation}\label{eq:transfer_rep_second_linear_proof}
\langle\psi_{\mathrm{AKLT}}^{(N)}|
(\mathbb{I}^{\otimes m}\otimes Y\otimes\mathbb{I}^{\otimes p})
|\psi_{\mathrm{AKLT}}^{(N)}\rangle
=
\Tr_{\mathrm{so}}\!\bigl(
\Phi^{p}\circ\widehat{Y}\circ\Phi^{m}
\bigr),
\end{equation}
where \(\Phi\) is the AKLT transfer channel defined in \eqref{eq:aklt_channel} and \(\Tr_{\mathrm{so}}\) denotes the superoperator trace from \eqref{eq:superop_trace_def}. This expresses the finite-volume expectation entirely in terms of the   maps \(\Phi\) and \(\widehat{Y}\) acting on the auxiliary algebra \(\mathcal{B}(\mathcal{H}_{\frac{1}{2}})\).

By Lemma~\ref{lem:aklt_conv}, the channel \(\Phi\) is primitive and bistochastic, with unique invariant density matrix \(\rho_*=\tfrac12\mathbb{I}_2\), and its iterates converge in operator norm to the idempotent quantum channel
\[
\Phi_{\infty}(M)
=
\frac{\Tr(M)}{2}\,\mathbb{I}_2,
\qquad M\in\mathbb{M}_2(\mathbb{C}).
\]
Fix \(p\in\mathbb{N}\) and let \(m\to\infty\) in \eqref{eq:transfer_rep_second_linear_proof}. Continuity of composition in finite dimension implies that
\[
\lim_{m\to\infty}
\Phi^{p}\circ\widehat{Y}\circ\Phi^{m}(M)
=
\Phi^{p}\circ\widehat{Y}\circ\Phi_{\infty}(M),
\qquad M\in\mathbb{M}_2(\mathbb{C}).
\]
Now let \(p\to\infty\). Using again Lemma~\ref{lem:aklt_conv} and the fact that composition is continuous, we obtain
\[
\Phi^{p}\circ\widehat{Y}\circ\Phi_{\infty}(M)
\longrightarrow
\Phi_{\infty}\circ\widehat{Y}\circ\Phi_{\infty}(M),
\]
and the convergence of \(\Phi^{L}\) being uniform on the unit ball, the double limit
\[
\lim_{\substack{m\to\infty\\ p\to\infty}}\Phi^{p}\circ\widehat{Y}\circ\Phi^{m}(M)
\]
exists, is independent of the order in which \(m\) and \(p\) tend to infinity, and equals \(\Phi_{\infty}\circ\widehat{Y}\circ\Phi_{\infty}(M)\) for every \(M\in\mathbb{M}_2(\mathbb{C})\).
The explicit form of \(\Phi_{\infty}\) allows us to compute this composition in closed form. For each \(M\),
\[
\Phi_{\infty}\circ\widehat{Y}\circ\Phi_{\infty}(M)
=
\Phi_{\infty}\bigl(\widehat{Y}(\tfrac{\Tr(M)}{2}\mathbb{I}_2)\bigr)
=
\frac{\Tr(M)}{2}\,\Phi_{\infty}\bigl(\widehat{Y}(\mathbb{I}_2)\bigr)
=
\frac{\Tr(M)}{2}\cdot\frac{\Tr(\widehat{Y}(\mathbb{I}_2))}{2}\,\mathbb{I}_2
=
\frac{1}{4}\Tr(\widehat{Y}(\mathbb{I}_2))\,\Tr(M)\,\mathbb{I}_2.
\]
Thus the sequence of superoperators \(\Phi^{p}\circ\widehat{Y}\circ\Phi^{m}\) converges in norm to the rank-one map
\[
M\longmapsto\frac{1}{4}\Tr(\widehat{Y}(\mathbb{I}_2))\,\Tr(M)\,\mathbb{I}_2.
\]

Since \(\Tr_{\mathrm{so}}\) is continuous on the finite-dimensional space of superoperators, we may pass to the limit in \eqref{eq:transfer_rep_second_linear_proof} and define
\[
\omega(Y)
:=
\lim_{\substack{m\to\infty\\ p\to\infty}}
\langle\psi_{\mathrm{AKLT}}^{(N)}|
(\mathbb{I}^{\otimes m}\otimes Y\otimes\mathbb{I}^{\otimes p})
|\psi_{\mathrm{AKLT}}^{(N)}\rangle
=
\Tr_{\mathrm{so}}\!\bigl(\Phi_{\infty}\circ\widehat{Y}\circ\Phi_{\infty}\bigr).
\]
Evaluating the superoperator trace using \eqref{eq:superop_trace_def} and the preceding expression gives
\[
\Tr_{\mathrm{so}}\!\bigl(\Phi_{\infty}\circ\widehat{Y}\circ\Phi_{\infty}\bigr)
=
\frac{1}{2}\,\Tr\!\bigl(\widehat{Y}(\mathbb{I}_2)\bigr),
\]
so we arrive at the compact formula
\[
\omega(Y)=\frac{1}{2}\,\Tr\!\bigl(\widehat{Y}(\mathbb{I}_2)\bigr).
\]

The claimed trace expression \eqref{eq:omega_infinite_formula_second_linear} now follows by inserting the Kraus form \eqref{eq:hatY_def_lemma} of \(\widehat{Y}\) into \(\Tr(\widehat{Y}(\mathbb{I}_2))\). Indeed,
\[
\widehat{Y}(\mathbb{I}_2)
=
\sum_{k,\ell}
\langle k_1,\dots,k_n|Y|\ell_1,\dots,\ell_n\rangle\,
A_{k_1}\cdots A_{k_n}A_{\ell_n}^{\dagger}\cdots A_{\ell_1}^{\dagger},
\]
so
\[
\Tr\!\bigl(\widehat{Y}(\mathbb{I}_2)\bigr)
=
\sum_{k,\ell}
\langle k_1,\dots,k_n|Y|\ell_1,\dots,\ell_n\rangle\,
\Tr\!\bigl(A_{k_1}\cdots A_{k_n}A_{\ell_n}^{\dagger}\cdots A_{\ell_1}^{\dagger}\bigr),
\]
and substituting into \(\omega(Y)=\tfrac{1}{2}\Tr(\widehat{Y}(\mathbb{I}_2))\) yields \eqref{eq:omega_infinite_formula_second_linear}.

It remains to verify normalisation. For \(Y=\mathbb{I}_{\mathcal{H}_{1}}^{\otimes n}\) we have \(\langle k|Y|\ell\rangle=\delta_{k,\ell}\) and
\[
\omega(\mathbb{I}_{\mathcal{H}_{1}}^{\otimes n})
=
\frac{1}{2}\sum_{k_1,\dots,k_n}
\Tr\!\bigl(A_{k_1}\cdots A_{k_n}A_{k_n}^{\dagger}\cdots A_{k_1}^{\dagger}\bigr).
\]
Define
\[
X_n
:=
\sum_{k_1,\dots,k_n}
A_{k_1}\cdots A_{k_n}A_{k_n}^{\dagger}\cdots A_{k_1}^{\dagger}
\]
The AKLT gauge relation \(\sum_k A_k^{\dagger}A_k=\mathbb{I}_2\) implies by a straightforward induction that \(X_n=\mathbb{I}_2\) for all \(n\in\mathbb{N}\), and therefore
\[
\omega(\mathbb{I}_{\mathcal{H}_{1}}^{\otimes n})
=
\frac{1}{2}\Tr(X_n)
=
\frac{1}{2}\Tr(\mathbb{I}_2)
=
1.
\]
Thus \(\omega\) is a normalised positive functional on the local algebra, and by consistency extends uniquely to a state on the quasi-local tensor product \(\mathcal{B}(\mathcal{H}_{1})^{\otimes\mathbb{N}}\), as claimed.
\end{proof}

The extension of $\omega$ from local observables to the quasi-local algebra and the purity of $\omega$ follow from the general theory of injective finitely correlated states , applied to the triple constructed below.
To make the connection with the general FCS formalism precise, we recast the AKLT MPS as a C$^{*}$-finitely correlated state. The auxiliary algebra is chosen as $\mathcal{B} = \mathbb{M}_2(\mathbb{C})\cong\mathcal{B}(\mathcal{H}_{\frac{1}{2}})$, reflecting the bond dimension of the MPS. The distinguished element $e\in\mathcal{B}$ and the reference functional $\rho:\mathcal{B}\to\mathbb{C}$ are taken to be
\[
e = \mathbb{I}_2,
\qquad
\rho(B) = \Tr(B),
\quad
B\in\mathcal{B}.
\]
The physical observable algebra is $\mathcal{A}=\mathcal{B}(\mathcal{H}_{1})\cong\mathbb{M}_3(\mathbb{C})$. We define a transition expectation $\mathbb{E}:\mathcal{A}\otimes\mathcal{B}\to\mathcal{B}$ by
\begin{equation}\label{eq:E_def_AKLT}
\mathbb{E}(X\otimes Y)
=
\sum_{k,k'\in\{+,0,-\}}
\langle k'|Y|k\rangle\,A_k^{\dagger} X A_{k'},
\qquad
X\in\mathcal{B},\ Y\in\mathcal{A}.
\end{equation}
Complete positivity follows from the Kraus form, and the normalization
\[
\mathbb{E}(\mathbb{I}_{\mathcal{B}}\otimes\mathbb{I}_{\mathcal{A}})
=
\sum_{k} A_k^{\dagger} A_k
=
\mathbb{I}_2
\]
realises the FCS condition $\mathbb{E}(\mathbb{I}_{\mathcal{B}}\otimes e)=e$, which is equivalent to the standard MPS gauge constraint $\sum_k A_k^{\dagger}A_k=\mathbb{I}_2$ . For each local observable $Y\in\mathcal{A}$, we write $\mathbb{E}_Y:\mathcal{B}\to\mathcal{B}$ for the associated completely positive map
\[
\mathbb{E}_Y(X)
=
\mathbb{E}(X\otimes Y)
=
\sum_{k,k'} \langle k'|Y|k\rangle\,A_k^{\dagger} X A_{k'}.
\]
In particular, $\mathbb{E}_{\mathbb{I}_{\mathcal{A}}} = \Phi^{*}$ is the Heisenberg dual of the transfer channel.

With this notation, the expectation value of a local product observable $Y_1\otimes\cdots\otimes Y_m\in\mathcal{A}^{\otimes m}$ in the infinite-volume AKLT state can be written in the standard FCS form
\begin{equation}\label{eq:AKLT_FCS_expectation}
\omega_{\mathrm{AKLT}}(Y_1\otimes\cdots\otimes Y_m)
=
\rho\!\left(
\mathbb{E}_{Y_1}\circ\cdots\circ\mathbb{E}_{Y_m}(e)
\right)
=
\Tr\!\left(
\mathbb{E}_{Y_1}\circ\cdots\circ\mathbb{E}_{Y_m}(\mathbb{I}_2)
\right).
\end{equation}
Forumulae \eqref{eq:AKLT_FCS_expectation} shows that all local correlators are generated by a finite-dimensional quantum Markovian memory living in $\mathcal{B}$, propagated by the completely positive maps $\mathbb{E}_Y$ and “read out’’ by the trace functional $\rho$. In particular, the spectral properties of the transfer channel $\Phi$ control the decay of correlations and the emergence of a mass gap, in agreement with the general theory of finitely correlated states and their ergodic and mixing behaviour \cite{FNW92}.

This algebraic perspective clarifies the thermodynamic limit construction: the sequence $(\omega_n)_{n\in\mathbb{N}}$ obtained from the finite-volume MPS is replaced by a single state $\omega$ on the quasi-local algebra, completely determined by the FCS triple $(\mathcal{B},\mathbb{E},\rho)$. The AKLT model thus appears as a particularly simple, yet highly non-trivial, instance in which topological features (edge modes, string order, and symmetry-protected topological structure) coexist with an exactly solvable transfer channel and a finite-dimensional auxiliary system, providing an ideal testing ground for the interplay between entanglement, symmetry, and quantum Markovianity in one-dimensional quantum spin systems.

\section{Causal HQMM Realization of the AKLT State}\label{sec:HQMM_AKLT}

In this section we show in a fully explicit and rigorous way that the infinite-volume AKLT state constructed in Section~\ref{thm:AKLT_infinite_volume} is exactly the observation process of a causal  HQMM . The hidden system will be the virtual spin-\(\tfrac12\) space of the AKLT matrix product representation, and the emission mechanism will be generated by the AKLT tensors. This provides a concrete example of the general correspondence between matrix product states and HQMMs.

We keep all notation from the previous section: the local physical space is the spin‑1 Hilbert space \(\mathcal{H}_{1}\), the virtual (hidden) space is \(\mathcal{H}\cong\mathbb{C}^2\), the corresponding observable algebras are \(\mathcal{B}_O:=\mathcal{B}(\mathcal{H}_{1})\) and \(\mathcal{B}_H:=\mathcal{B}(\mathcal{H}_{\frac{1}{2}})\), and the AKLT tensors \(\{A_k\}_{k\in\{+,0,-\}}\subset\mathcal{B}_H\) and infinite-volume state \(\omega_{\mathrm{AKLT}}\) are defined as in Theorem~\ref{thm:AKLT_infinite_volume}.

To connect this structure to HQMMs, we first define the emission expectation, which encodes how a hidden operator and a physical observable combine into a new hidden observable. We set
\begin{equation}\label{eq:EOH_def_HQMM}
\mathcal{E}_{O,H}:\mathcal{B}_H\otimes\mathcal{B}_O\to\mathcal{B}_H,
\qquad
\mathcal{E}_{O,H}(X\otimes Y)
:=
\sum_{k,\ell\in\{+,0,-\}}\langle k|Y|\ell\rangle\,A_k X A_{\ell}^{\dagger},
\end{equation}
for \(X\in\mathcal{B}_H\), \(Y\in\mathcal{B}_O\), where \(\{|+\rangle,|0\rangle,|-\rangle\}\) is the fixed orthonormal basis of \(\mathcal{H}_{1}\). The Kraus form shows immediately that \(\mathcal{E}_{O,H}\) is completely positive. Unitality follows from the AKLT gauge constraint: using \(\sum_k|k\rangle\langle k|=\mathbb{I}_{\mathcal{H}_{1}}\) and \(\sum_k A_kA_k^{\dagger}=\mathbb{I}_{\mathcal{H}}\) we obtain
\[
\mathcal{E}_{O,H}(\mathbb{I}_2\otimes\mathbb{I}_{\mathcal{H}_{1}})
=
\sum_{k}A_k A_k^{\dagger}
=
\mathbb{I}_2,
\]
so \(\mathcal{E}_{O,H}\) is an emission expectation (see Appendix\ref{app:A} for a rigrous definition of HQMMs). This map is the natural HQMM analogue of the “transfer operator with insertion” used to compute expectation values in the AKLT finitely correlated state \cite{FNW92}.

 \begin{theorem}
\label{thm:causal-state-kraus}
Let $\Xi_{\mathrm{caus}} = \bigl(\phi_{H,1}; \{\mathcal{E}_{H;m}\}_{m=1}^{n}, \{\mathcal{E}_{H,O;m}\}_{m=1}^{n}, \{\mathcal{G}^{(m)}\}_{m=1}^{n}\bigr)$ be a causal HQMM defined on $\bigl(\mathcal{B}(\mathcal{H}_{\frac{1}{2}})\otimes \mathcal{B}(\mathcal{H}_{1})\bigr)^{\otimes \mathbb{N}}$. If the emission transition expectation $\mathcal{E}_{H,O;m}$ is given by \textnormal{(\ref{eq:EOH_def_HQMM})}, then the infinite-volume state $\Psi_{H,O}$ defined in \textnormal{\eqref{eq:caus-state-no-sup}} admits the explicit representation
\begin{align}
\Psi_{H,O}\!\left(\bigotimes_{m=1}^n (a_m\otimes b_m)\right)
&=
\sum_{\substack{k_1,\dots,k_n\\ \ell_1,\dots,\ell_n}}
\left(\prod_{m=1}^n \langle k_m|b_m|\ell_m\rangle\right)
\nonumber\\[0.3em]
&\quad\times
\phi_{H,1}\!\Bigl(
A_{k_1}\,
\mathcal{E}_{H;1}\Bigl(
a_1\,
A_{k_2}\,
\mathcal{E}_{H;2}\Bigl(
a_2\,
\cdots
A_{k_n}\,\mathcal{E}_{H;n}(a_n)\,
A_{\ell_n}^\dagger
\cdots
\Bigr)
A_{\ell_2}^\dagger
\Bigr)
A_{\ell_1}^\dagger
\Bigr).
\label{eq:psi-explicit-compact2}
\end{align}
for every $n\in\mathbb{N}$ and all $a_m\in \mathcal{B}(\mathcal{H}_{\frac{1}{2}})$, $b_m\in \mathcal{B}(\mathcal{H}_{1})$, $1\le m\le n$.
\end{theorem}
\begin{proof}[\textbf{Proof of Theorem \ref{thm:causal-state-kraus}}]
For each \(k = 0,1,\dots,n\), define the operator
\[
T_k := \mathcal{G}^{(n-k+1)}_{a_{n-k+1},b_{n-k+1}} \circ \cdots \circ \mathcal{G}^{(n)}_{a_n,b_n} (\mathbb{I}_{H;n+1}),
\]
with \(T_0 := \mathbb{I}_{H;n+1}\).  Clearly \(T_n = \mathcal{G}^{(1)}_{a_1,b_1}\circ\cdots\circ\mathcal{G}^{(n)}_{a_n,b_n}(\mathbb{I}_{H;n+1})\).

We prove by induction on \(k\) that
\[
T_k = \sum_{\substack{k_{n-k+1},\dots,k_n\\ \ell_{n-k+1},\dots,\ell_n}}
\Bigl(\prod_{m=n-k+1}^{n} \langle k_m|b_m|\ell_m\rangle\Bigr)
\,
A_{k_{n-k+1}}\,
\mathcal{E}_{H;n-k+1}\Bigl(
a_{n-k+1},\,
\cdots
A_{k_n}\,\mathcal{E}_{H;n}(a_n)\,A_{\ell_n}^\dagger
\cdots
\Bigr)
A_{\ell_{n-k+1}}^\dagger \tag{$\diamond$}
\]

For \(k = 0\), both sides equal \(\mathbb{I}_{H;n+1}\) by convention, establishing the base case.  Assume (1) holds for some \(k < n\).  Since \(T_{k+1} = \mathcal{G}^{(n-k)}_{a_{n-k},b_{n-k}}(T_k)\), the definition of \(\mathcal{G}^{(n-k)}\) together with the Kraus representation \eqref{eq:EOH_def_HQMM} gives
\[
T_{k+1} = \sum_{k_{n-k},\ell_{n-k}} \langle k_{n-k}|b_{n-k}|\ell_{n-k}\rangle \,
A_{k_{n-k}}\; \mathcal{E}_{H;n-k}\bigl( a_{n-k}\otimes T_k \bigr) \; A_{\ell_{n-k}}^\dagger
\]
Substituting the inductive expression for \(T_k\) and using linearity of \(\mathcal{E}_{H;n-k}\) yields exactly the right-hand side of ($\diamond$) with \(k\) replaced by \(k+1\).  Hence ($\diamond$) holds for all \(k = 0,\dots,n\) by induction.

Setting \(k = n\) in (1) provides an explicit formula for \(T_n\).  Applying the initial state \(\phi_{H,1}\) and the definition
\[
\Psi_{H,O}\Bigl(\bigotimes_{m=1}^n (a_m\otimes b_m)\Bigr) = \phi_{H,1}(T_n)
\]
produces precisely \eqref{eq:psi-explicit-compact2}. 
\end{proof}

In the AKLT-inspired HQMM framework of \cite{SA26}, the hidden space is taken as a qubit Hilbert space $\mathcal{H}\cong\mathbb{C}^2$ with orthonormal basis $\{|\uparrow\rangle,|\downarrow\rangle\}$, and the hidden transition expectation
\(
\mathcal{E}_{H} : \mathcal{B}(\mathcal{H}_{\frac{1}{2}}\otimes\mathcal{H}_{\frac{1}{2}})\to\mathcal{B}(\mathcal{H}_{\frac{1}{2}})
\)
is implemented via a partial isometry
\[
V:\mathcal{H}_{\frac{1}{2}}\longrightarrow\mathcal{H}\otimes\mathcal{H}_{\frac{1}{2}},
\qquad
\mathcal{E}_{H}(X)=V^\dagger X V,
\quad X\in\mathcal{B}(\mathcal{H}\otimes\mathcal{H}_{\frac{1}{2}}).
\]
To capture the AKLT valence-bond structure, one chooses
\[
\psi^-=\frac{1}{\sqrt{2}}\bigl(|\uparrow\downarrow\rangle-|\downarrow\uparrow\rangle\bigr),
\qquad
\psi^+=\frac{1}{\sqrt{2}}\bigl(|\uparrow\downarrow\rangle+|\downarrow\uparrow\rangle\bigr),
\]
and defines
\[
V|\uparrow\rangle=\psi^-,
\qquad
V|\downarrow\rangle=\psi^+,
\]
so that $V$ is an isometry ($V^\dagger V=\mathbb{I}_{\mathcal{H}}$) and its range is the subspace spanned by a maximally entangled singlet and a symmetric triplet, thereby encoding the AKLT virtual bonds in the hidden transition itself.

In our notation, the emission transition expectation
\(
\mathcal{E}_{H,O}:\mathcal{B}(\mathcal{H}_{\frac{1}{2}})\otimes\mathcal{B}(\mathcal{H}_{1})\to\mathcal{B}(\mathcal{H}_{\frac{1}{2}})
\)
admits the Kraus representation
\[
\mathcal{E}_{H,O}(x\otimes y)
=
\sum_{k,\ell} A_k\, x\, A_\ell^\dagger\,\langle k\mid y\mid\ell\rangle,
\qquad x\in\mathcal{B}(\mathcal{H}_{\frac{1}{2}}),\,y\in\mathcal{B}(\mathcal{H}_{1}),
\]
for the orthonormal basis  $\{|+\rangle,|0\rangle,|-\rangle\}$  of $\mathcal{H}_{1}$ and matrix units $\langle k\mid y\mid\ell\rangle$.

The \emph{conventional} HQMM “black-box’’ update associated with a pair $(a,b)\in\mathcal{B}(\mathcal{H}_{\frac{1}{2}})\times\mathcal{B}(\mathcal{H}_{1})$ is then
\[
\mathcal{F}_{a,b} : \mathcal{B}(\mathcal{H}_{\frac{1}{2}})\to\mathcal{B}(\mathcal{H}_{\frac{1}{2}}),
\qquad
\mathcal{F}_{a,b}(x)
=
\mathcal{E}_{H}\!\Bigl(
  \mathcal{E}_{H,O}\bigl(a\otimes b\bigr)
  \otimes x
\Bigr),
\]
and, with the above Kraus form and choice of $V$, can be written explicitly as
\begin{align*}
\mathcal{F}_{a,b}(x)
&=
V^\dagger\!\Biggl(
  \sum_{k,\ell} A_k\, a\, A_\ell^\dagger\,\langle k\mid y\mid\ell\rangle
  \;\otimes\; x
\Biggr)V\\
&=
\sum_{k,\ell}
\langle k\mid y\mid\ell\rangle
V^\dagger\bigl(A_k\, a\, A_\ell^\dagger\otimes x\bigr)V
\end{align*}
Here the entangling structure of $V$ enters only through the conjugation $V^\dagger(\cdot)V$, so $\mathcal{F}_{a,b}$ is a single CPTP map on $\mathcal{B}(\mathcal{H}_{\frac{1}{2}})$ in which the internal tensor structure of $\mathcal{E}_{H}$ is already “collapsed’’. By contrast, in the \emph{causal} HQMM architecture the basic block map associated with $(a,b)$ acts on $\mathcal{B}(\mathcal{H}_{\frac{1}{2}})$ as
\[
\mathcal{G}_{a,b}(x)
=
\mathcal{E}_{H,O}\bigl(\mathcal{E}_{H}(x\otimes a)\otimes b\bigr),
\qquad x\in\mathcal{B}(\mathcal{H}_{\frac{1}{2}}),
\]
so that the hidden entangling step \(\mathcal{E}_{H}\) precedes the emission. Using the explicit choice of $V$ and the Kraus form of $\mathcal{E}_{H,O}$, we obtain
\[
\mathcal{G}_{a,b}(x)
=
\sum_{k,\ell}
A_k\,
\Bigl[\,V^\dagger\bigl(x\otimes a\bigr)V\,\Bigr]\,
A_\ell^\dagger\;\langle k,b\,\ell\rangle
=
\sum_{k,\ell}
\langle k,b\,\ell\rangle\;
A_k\,V^\dagger\bigl(x\otimes a\bigr)V\,A_\ell^\dagger.
\]
Here the hidden entanglement created by the map $x\mapsto V^\dagger(x\otimes a)V$ remains visible \emph{inside} the Kraus expansion at each step, and the sequential composition of such $\mathcal{G}_{a,b}$–maps produces the full nested expression of the causal state $\Psi_{H,O}$.

Thus, for the same AKLT-type choice of $V$, the conventional HQMM encodes the hidden entanglement only through the \emph{effective} channel $\mathcal{F}_{a,b}$, whereas the causal HQMM retains the full tensor-product and entangling structure of $\mathcal{E}_{H}$ under the $\mathcal{G}_{a,b}$ maps. This structural difference indicates that the two models need not be stochastically equivalent; in particular, for AKLT-type choices of $V$ one expects that the family of observable processes generated by iterates of $(\mathcal{G}_{a,b})_{a,b}$ is strictly richer than that generated by $(\mathcal{F}_{a,b})_{a,b}$, a question that can be investigated rigorously in future work.

In the AKLT–HQMM realization we fix the hidden algebra as $\mathcal{B}_H=\mathcal{B}(\mathcal{H}_{\frac{1}{2}})\cong\mathbb{M}_2(\mathbb{C})$ and specify the hidden transition expectation as the rank–one, unital, completely positive map
\begin{equation}\label{eq:E_H_def}
\mathcal{E}_{H}:\mathcal{B}_H\otimes\mathcal{B}_H\longrightarrow\mathcal{B}_H,
\qquad
\mathcal{E}_{H}(X\otimes Z)
=
\tfrac{1}{2}\,\operatorname{Tr}(X)\,Z.
\end{equation}
Thus $\mathcal{E}_H$ realises the canonical conditional expectation onto $\mathbb{C}\mathbb{I}_2\subset\mathcal{B}_H$ in the first tensor factor, followed by the identity on the second. Equivalently, its dual channel maps every normal state on $\mathcal{B}_H$ to the maximally mixed state $\rho_*=\tfrac{1}{2}\mathbb{I}_2$, so that the hidden dynamics is maximally mixing in a single step. From the MPS perspective, any primitive unital channel with faithful invariant state $\rho_*$ is equivalent, up to a gauge transformation on the virtual space, to this choice; \eqref{eq:E_H_def} is therefore a convenient normal form in which the entire nontrivial structure of the AKLT chain is encoded in the emission part.

The emission transition expectation
\[
\mathcal{E}_{H,O}:\mathcal{B}_H\otimes\mathcal{B}_O\to\mathcal{B}_H
\]
is fixed by \eqref{eq:EOH_def_HQMM} in terms of the AKLT Kraus operators $\{A_k\}_{k\in\{+,0,-\}}$. For each site $t$ and each pair $(X_t,Y_t)\in\mathcal{B}_H\times\mathcal{B}_O$ we obtain a Heisenberg one–step block map on the hidden algebra
\begin{equation}\label{eq:G_block_def-prep}
\mathcal{G}_{X_t,Y_t}:\mathcal{B}_H\to\mathcal{B}_H,
\qquad
\mathcal{G}_{X_t,Y_t}(Z)
=
\mathcal{E}_{H,O}\bigl(\mathcal{E}_H(X_t\otimes Z)\otimes Y_t\bigr),
\end{equation}
which, after inserting \eqref{eq:E_H_def} and \eqref{eq:EOH_def_HQMM}, takes the explicit Kraus form
\begin{equation}\label{eq:G_single_step_def-prep}
\mathcal{G}_{X_t,Y_t}(Z)
=
\tfrac{1}{2}\,\operatorname{Tr}(X_t)
\sum_{k_t,\ell_t\in\{+,0,-\}}
\langle \ell_t|Y_t|k_t\rangle\,
A_{k_t}\,Z\,A_{\ell_t}^{\dagger},
\qquad Z\in\mathcal{B}_H.
\end{equation}
In particular, the family $\{\mathcal{G}_{X_t,Y_t}\}_{t\ge 1}$ realises a completely positive “transfer operator’’ on the virtual space whose Kraus operators are precisely the AKLT tensors, weighted by the matrix elements of the local observables. Iterated compositions of these maps will implement the causal HQMM dynamics and, after suitable boundary choices and restriction to $\mathcal{B}_O^{\otimes n}$, will produce a functional that can be compared directly with the AKLT finitely correlated state of Theorem~\ref{thm:AKLT_infinite_volume}.

\begin{theorem}\label{thm:AKLT_is_observation_rigorous}
Let \(\mathcal{B}_H=\mathcal{B}(\mathcal{H}_{\frac{1}{2}})\cong\mathbb{M}_2(\mathbb{C})\) and \(\mathcal{B}_O=\mathcal{B}(\mathcal{H}_{1})\cong\mathbb{M}_3(\mathbb{C})\), and define the hidden transition expectation \(\mathcal{E}_H\) by \eqref{eq:E_H_def} and the emission expectation \(\mathcal{E}_{O,H}\) by \eqref{eq:EOH_def_HQMM} in terms of the AKLT tensors \(\{A_k\}\) from Section~\ref{thm:AKLT_infinite_volume}. Consider the causal HQMM whose one-step block maps \(\mathcal{G}_{X_t,Y_t}\) are given by \eqref{eq:G_single_step_def-prep} and whose initial hidden state and boundary operator are chosen as \(\phi_{H,0}=\Tr\) and \(h=\mathbb{I}_2\). Then, for every local observable \(Y\), the observation process satisfies
\begin{equation}\label{eq:omega_AKLT_MPS_factorized}
\psi_{O}(Y)=\omega(Y)
=
\frac{1}{2}\sum_{\substack{k_1,\dots,k_n\\ \ell_1,\dots,\ell_n}}
\Bigl(\prod_{m=1}^{n}\langle \ell_m|b_m|k_m\rangle\Bigr)
\Tr\!\bigl(
 A_{k_1}\cdots A_{k_n}A_{\ell_n}^{\dagger}\cdots A_{\ell_1}^{\dagger}
\bigr).
\end{equation}
where \(\omega_{\mathrm{AKLT}}\) is the infinite-volume AKLT state. In particular, the AKLT state is the observation process of a causalHQMM.
\end{theorem}
\begin{proof}
We work in the setting of Theorems~\ref{thm:AKLT_infinite_volume} and \ref{thm:causal-state-kraus}. The hidden algebra is $\mathcal{B}_H=\mathcal{B}(\mathcal{H}_{\frac{1}{2}})\cong\mathbb{M}_2(\mathbb{C})$ and the observable algebra is $\mathcal{B}_O=\mathcal{B}(\mathcal{H}_{1})\cong\mathbb{M}_3(\mathbb{C})$, where $\mathcal{H}_{1}$ carries the orthonormal basis $\{|+\rangle,|0\rangle,|-\rangle\}$ and the AKLT tensors $\{A_k\}_{k\in\{+,0,-\}}$ from \eqref{eq:Ak} and Theorem~\ref{thm:AKLT_infinite_volume}.  The hidden transition expectation $\mathcal{E}_H$ and the emission expectation $\mathcal{E}_{H,O}$ are given by \eqref{eq:E_H_def} and \eqref{eq:EOH_def_HQMM}, and we choose as initial hidden state the normalised trace
\[
\phi_{H,1}=\tfrac{1}{2}\Tr:\mathcal{B}_H\to\mathbb{C},
\]
together with the boundary operator $h=\mathbb{I}_2$. With these data we obtain a causal HQMM
\[
\Xi_{\mathrm{caus}}
=
\bigl(\phi_{H,1};\{\mathcal{E}_{H;m}\}_{m=1}^{n},\{\mathcal{E}_{H,O;m}\}_{m=1}^{n},\{\mathcal{G}^{(m)}\}_{m=1}^{n}\bigr)
\]
in the sense of Theorem~\ref{thm:causal-state-kraus}, with $\mathcal{E}_{H;m}=\mathcal{E}_H$ and $\mathcal{E}_{H,O;m}=\mathcal{E}_{H,O}$ for all $m$.

For any $n\in\mathbb{N}$ and factorised input $\bigotimes_{m=1}^n(a_m\otimes b_m)$, with $a_m\in\mathcal{B}_H$ and $b_m\in\mathcal{B}_O$, Theorem~\ref{thm:causal-state-kraus} yields the Kraus representation
\begin{align}
\Psi_{H,O}\!\left(\bigotimes_{m=1}^n (a_m\otimes b_m)\right)
&=
\sum_{\substack{k_1,\dots,k_n\\ \ell_1,\dots,\ell_n}}
\left(\prod_{m=1}^n \langle k_m|b_m|\ell_m\rangle\right)
\nonumber\\[0.3em]
&\quad\times
\phi_{H,1}\!\Bigl(
A_{k_1}\,
\mathcal{E}_{H;1}\Bigl(
a_1\,
A_{k_2}\,
\mathcal{E}_{H;2}\Bigl(
a_2\,
\cdots
A_{k_n}\,\mathcal{E}_{H;n}(a_n)\,
A_{\ell_n}^\dagger
\cdots
\Bigr)
A_{\ell_2}^\dagger
\Bigr)
A_{\ell_1}^\dagger
\Bigr)
\label{eq:psi-explicit-compact2-proof}
\end{align}
where all $\mathcal{E}_{H;m}$ are identified with $\mathcal{E}_H$.

Since $\mathcal{E}_H$ has the product form $\mathcal{E}_H(a_m\otimes x)=\tfrac{1}{2}\Tr(a_m)\,x$ and $\phi_{H,1}=\tfrac{1}{2}\Tr$, the nested term inside $\phi_{H,1}$ can be evaluated explicitly by iterating this identity from $m=1$ up to $m=n$. A direct computation gives
\begin{align}
&\phi_{H,1}\!\Bigl(
A_{k_1}\,
\mathcal{E}_{H;1}\Bigl(
a_1\,
A_{k_2}\,
\mathcal{E}_{H;2}\Bigl(
a_2\,
\cdots
A_{k_n}\,\mathcal{E}_{H;n}(a_n\otimes\mathbb{I}_2)\,
A_{\ell_n}^\dagger
\cdots
\Bigr)
A_{\ell_2}^\dagger
\Bigr)
A_{\ell_1}^\dagger
\Bigr)\nonumber\\
&=
\frac{1}{2}\Tr\Bigl(
A_{k_1}\,
\mathcal{E}_{H;1}\bigl(\cdots\bigr)
A_{\ell_1}^\dagger
\Bigr)
\nonumber\\[0.3em]
&=
\frac{1}{2^2}\Tr(a_1)\,
\Tr\Bigl(
A_{k_1}A_{k_2}\,
\mathcal{E}_{H;2}\bigl(\cdots\bigr)
A_{\ell_2}^\dagger A_{\ell_1}^\dagger
\Bigr)
\nonumber\\
&\;\;\vdots\nonumber\\
&=
\frac{1}{2^{n+1}}
\Bigl(\prod_{m=1}^n \Tr(a_m)\Bigr)
\Tr\Bigl(
A_{k_1}\cdots A_{k_n}
A_{\ell_n}^\dagger\cdots A_{\ell_1}^\dagger
\Bigr)
\label{eq:phiH1-iteration-AKLT-clean}
\end{align}
where at each step $\mathcal{E}_{H;m}$ contributes a factor $\tfrac{1}{2}\Tr(a_m)$ and the final action of $\phi_{H,1}$ contributes the extra $\tfrac{1}{2}$.
Substituting \eqref{eq:phiH1-iteration-AKLT-clean} into \eqref{eq:psi-explicit-compact2-proof}, we obtain the closed Kraus expansion
\begin{align}
\Psi_{H,O}\!\left(\bigotimes_{m=1}^n (a_m\otimes b_m)\right)
&=
\tfrac{1}{2^{n+1}}
\sum_{\substack{k_1,\dots,k_n\\ \ell_1,\dots,\ell_n}}
\Bigl(\prod_{m=1}^{n}\Tr(a_m)\,\langle k_m|b_m|\ell_m\rangle\Bigr)
\nonumber\\[0.3em]
&\quad\times
\Tr\!\bigl(A_{k_1}\cdots A_{k_n} A_{\ell_n}^{\dagger}\cdots A_{\ell_1}^{\dagger}\bigr)
\label{eq:psi-HO-closed-iteration-clean}
\end{align}

The observation process is obtained by restricting to the observable tensor factors, that is,
\[
\psi_{O}(Y)
:=
\Psi_{H,O}(\mathbb{I}_2^{\otimes n}\otimes Y),
\qquad
Y=b_1\otimes\cdots\otimes b_n\in\mathcal{B}_O^{\otimes n}.
\]
Setting $a_m=\mathbb{I}_2$ in \eqref{eq:psi-HO-closed-iteration-clean} gives $\Tr(a_m)=2$ for all $m$ and hence
\[
\tfrac{1}{2^{n+1}}\prod_{m=1}^{n}\Tr(a_m)
=
\tfrac{1}{2^{n+1}}\,2^{n}
=
\tfrac{1}{2},
\]
so we obtain
\begin{equation}\label{eq:psi_O_closed_AKLT_final-clean}
\psi_{O}(Y)
=
\tfrac{1}{2}\sum_{\substack{k_1,\dots,k_n\\ \ell_1,\dots,\ell_n}}
\Bigl(\prod_{m=1}^{n}\langle k_m|b_m|\ell_m\rangle\Bigr)
\Tr\!\bigl(A_{k_1}\cdots A_{k_n} A_{\ell_n}^{\dagger}\cdots A_{\ell_1}^{\dagger}\bigr).
\end{equation}

By Theorem~\ref{thm:AKLT_infinite_volume}, the infinite-volume AKLT state $\omega=\omega_{\mathrm{AKLT}}$ is, for any local $Y\in\mathcal{B}(\mathcal{H}_{1})^{\otimes n}$, given by the MPS formula \eqref{eq:omega_infinite_formula_second_linear}. For factorised $Y=b_1\otimes\cdots\otimes b_n$, the matrix element factorises as
\[
\langle k_1,\dots,k_n|Y|\ell_1,\dots,\ell_n\rangle
=
\prod_{m=1}^n \langle k_m|b_m|\ell_m\rangle
=
\prod_{m=1}^n \langle \ell_m|b_m|k_m\rangle,
\]
so \eqref{eq:omega_infinite_formula_second_linear} reduces to exactly the same expression as the right-hand side of \eqref{eq:psi_O_closed_AKLT_final-clean}. Hence
\[
\psi_{O}(Y)=\omega(Y)
\]
for all factorised local observables $Y$, and therefore, by linearity and density, for all local $Y\in\mathcal{B}(\mathcal{H}_{1})^{\otimes n}$. The consistency of the quasi-local construction in Theorem~\ref{thm:AKLT_infinite_volume} then implies that this equality holds on the full quasi-local C\(^*\)-algebra $\mathcal{B}(\mathcal{H}_{1})^{\otimes\mathbb{N}}$.

We conclude that the infinite-volume AKLT finitely correlated state $\omega_{\mathrm{AKLT}}$ is exactly the observation process of the causal HQMM specified by \eqref{eq:E_H_def}, \eqref{eq:EOH_def_HQMM} and $\phi_{H,1}=\tfrac{1}{2}\Tr$. This proves the theorem.
\end{proof}

This theorem realises the AKLT finitely correlated state as the output of a finite-dimensional quantum Markovian memory—the virtual spin-\(\tfrac12\) system—driven by completely positive maps determined by the AKLT tensors.

\section{Discussion and Future Directions}\label{sect-Disc}
In this work, we have uncovered a HQMMs structure underlying the AKLT spin‑1 chain by showing that its FCS representation can be realised as the observation process of a suitably constructed causal HQMM  \cite{SouBarhqmm2026}. Meanwhile, the AKLT state cannot be realized as the observation process within either the conventional HQMM architecture \cite{SA26} or the EHMM framework \cite{Sou25}. This highlights that causal HQMMs possess a richer and more flexible structure, enabling the representation of quantum states with nontrivial entanglement patterns and topological order.

These findings point to several concrete directions for future work. One is to systematically investigate the general correspondence between HQMMs and tensor-network states, extending the MPS–EHMM result to broader hierarchies of HQMMs and to finitely correlated states used to model topological and SPT phases. Another is to classify when two HQMMs that generate the same observable state are equivalent in an appropriate stochastic or quasi-local sense, and how invariants such as correlation length, edge modes, and SPT indices are reflected at the hidden level. More broadly, our results suggest that faithfully representing SPT phases within hidden Markovian frameworks may require models with non-local memory, symmetry-twisted transition expectations, or explicitly topological hidden spaces, and they open a path towards understanding how far an extended stochastic formalism can go in capturing topological order before genuinely quantum many-body features become irreducible.


%

\appendix
\section{Appendix: Deferred Proofs}\label{app:B}

We provide the proofs of some technical lemmas for completeness.


\begin{proof}[\textbf{Proof of Lemma\ref{lem:aklt_conv}}]
Using the explicit form of the AKLT matrices $A_k$ in \eqref{eq:Ak}, one verifies by direct computation that the Pauli matrices are eigenoperators of $\Phi$: for each $\alpha\in\{x,y,z\}$,
\[
\Phi(\sigma_\alpha)
=
\sum_k A_k \sigma_\alpha A_k^{\dagger}
=
-\tfrac{1}{3}\,\sigma_\alpha,
\]
while $\Phi(\mathbb{I}_2)=\mathbb{I}_2$ as already observed. Thus the set $\{\mathbb{I}_2,\sigma_x,\sigma_y,\sigma_z\}$ forms an eigenbasis of $\Phi$ as a linear map on $\mathbb{M}_2(\mathbb{C})$, with eigenvalues $1$ on $\mathbb{I}_2$ and $-1/3$ on the three Pauli directions.

Every matrix $M\in\mathbb{M}_2(\mathbb{C})$ admits a unique decomposition in this basis,
\[
M
=
\tfrac{\operatorname{Tr}(M)}{2}\,\mathbb{I}_2
+
\sum_{\alpha\in\{x,y,z\}} c_{\alpha}\,\sigma_{\alpha},
\]
for suitable coefficients $c_{\alpha}\in\mathbb{C}$. Applying $\Phi^{n}$ and using the eigen-relations just established, we find
\[
\Phi^{n}(M)
=
\tfrac{\operatorname{Tr}(M)}{2}\,\Phi^{n}(\mathbb{I}_2)
+
\sum_{\alpha} c_{\alpha}\,\Phi^{n}(\sigma_{\alpha})
=
\tfrac{\operatorname{Tr}(M)}{2}\,\mathbb{I}_2
+
\sum_{\alpha} c_{\alpha}\,\bigl(-\tfrac{1}{3}\bigr)^{n}\sigma_{\alpha}
\]
for all $n\in\mathbb{N}$. Since $|-1/3|<1$, the second term converges to $0$ in operator norm as $n\to\infty$, which proves the limit \eqref{eq:Phi_power_limit}. The modulus of the sub-leading eigenvalue is $1/3$, so the convergence is exponential with rate $1/3$.
\end{proof}


\begin{proof}[\textbf{Proof of Lemma\ref{lem:hatY_properties}}]
 For (a),  using this expansion and the MPS form of \(|\psi_{\mathrm{AKLT}}^{(n)}\rangle\), the expectation value can be written as
\begin{align*}
\langle\psi_{\mathrm{AKLT}}^{(n)}|Y|\psi_{\mathrm{AKLT}}^{(n)}\rangle
&=
\sum_{k,\ell}
\langle k_1,\dots,k_n|Y|\ell_1,\dots,\ell_n\rangle\,
\operatorname{Tr}(A_{k_1}\cdots A_{k_n})\,
\overline{\operatorname{Tr}(A_{\ell_1}\cdots A_{\ell_n})},
\end{align*}
where we abbreviate multi-indices \(k=(k_1,\dots,k_n)\), \(\ell=(\ell_1,\dots,\ell_n)\). Using \(\overline{\operatorname{Tr}(X)}=\operatorname{Tr}(X^{\dagger})\), this becomes
\[
\langle\psi_{\mathrm{AKLT}}^{(n)}|Y|\psi_{\mathrm{AKLT}}^{(n)}\rangle
=
\sum_{k,\ell}
\langle k|Y|\ell\rangle\,
\operatorname{Tr}\bigl(A_{k_n}^{\dagger}\cdots A_{k_1}^{\dagger}\bigr)\operatorname{Tr}\bigl(A_{\ell_1}\cdots A_{\ell_n}\bigr).
\]
We now express the trace over \(\mathcal{H}\) in the basis \(\{|\alpha\rangle\}\),
\[
\operatorname{Tr}\bigl(A_{k_n}^{\dagger}\cdots A_{k_1}^{\dagger}\bigr)\operatorname{Tr}\bigl(A_{\ell_1}\cdots A_{\ell_n}\bigr)
=
\sum_{\alpha=1}^{2}
\langle\alpha|A_{k_n}^{\dagger}\cdots A_{k_1}^{\dagger}|\alpha\rangle\langle \beta A_{\ell_1}\cdots A_{\ell_n}|\alpha\rangle.
\]

On the other hand, by the definition \eqref{eq:hatY_def_lemma} of \(\widehat{Y}\), we have for each rank-one operator \(|\alpha\rangle\langle\beta|\),
\[
\widehat{Y}(|\alpha\rangle\langle\beta|)
=
\sum_{k,\ell}
\langle k|Y|\ell\rangle\,
A_{k_n}^{\dagger}\cdots A_{k_1}^{\dagger}|\alpha\rangle\langle\beta|A_{\ell_1}\cdots A_{\ell_n}.
\]
Taking the matrix element \(\langle\alpha|\cdot|\beta\rangle\) yields
\[
\langle\alpha|\widehat{Y}(|\alpha\rangle\langle\beta|)|\beta\rangle
=
\sum_{k,\ell}
\langle k|Y|\ell\rangle\,
\langle\alpha|A_{k_n}^{\dagger}\cdots A_{k_1}^{\dagger}|\alpha\rangle\,
\langle\beta|A_{\ell_1}\cdots A_{\ell_n}|\beta\rangle.
\]
Summing over \(\alpha,\beta\) and using completeness, we obtain
\begin{align*}
\operatorname{Tr}_{\mathrm{so}}(\widehat{Y})
&=
\sum_{\alpha,\beta}
\langle\alpha|\widehat{Y}(|\alpha\rangle\langle\beta|)|\beta\rangle \\
&=
\sum_{k,\ell}\langle k|Y|\ell\rangle
\sum_{\alpha,\beta}
\langle\alpha|A_{k_n}^{\dagger}\cdots A_{k_1}^{\dagger}|\alpha\rangle\,
\langle\beta|A_{\ell_1}\cdots A_{\ell_n}|\beta\rangle \\
&=
\sum_{k,\ell}\langle k|Y|\ell\rangle\,
\operatorname{Tr}\bigl(A_{k_n}^{\dagger}\cdots A_{k_1}^{\dagger}A_{\ell_1}\cdots A_{\ell_n}\bigr),
\end{align*}
which coincides with the expression for \(\langle\psi_{\mathrm{AKLT}}^{(n)}|Y|\psi_{\mathrm{AKLT}}^{(n)}\rangle\) above. This proves \eqref{eq:AKLT_expectation_hatY_correct}.   For the special case \(Y=\mathbb{I}^{\otimes n}\), the matrix elements are \(\langle k|Y|\ell\rangle=\delta_{k,\ell}\), so \eqref{eq:hatY_def_lemma} simplifies to
\[
\widehat{\mathbb{I}^{\otimes n}}(M)
=
\sum_{k_1,\dots,k_n}
A_{k_n}^{\dagger}\cdots A_{k_1}^{\dagger}\,M\,A_{k_1}\cdots A_{k_n}.
\]
Using the Bistochasticity  of the single-step AKLT  transfer channel \(\Phi(Z)=\sum_k A_k Z A_k^{\dagger} = \sum_k A_k^{\dagger} Z A_k\), one checks inductively that
\[
\Phi^{n}(M)
=
\sum_{k_1,\dots,k_n}
A_{k_n}^{\dagger}\cdots A_{k_1}^{\dagger}\,M\,A_{k_1}\cdots A_{k_n},
\]
so \(\widehat{\mathbb{I}^{\otimes n}}=\Phi^{n}\). The equality
\[
\langle\psi_{\mathrm{AKLT}}^{(n)}|\psi_{\mathrm{AKLT}}^{(n)}\rangle
=
\operatorname{Tr}_{\mathrm{so}}(\Phi^{n})
\]
follows by applying \eqref{eq:AKLT_expectation_hatY_correct} with \(Y=\mathbb{I}^{\otimes n}\).

For (b), assume that \(Y\in\mathcal{B}(\mathcal{H}_{1})^{\otimes n}\) acts on sites \(1,\dots,n\) and \(Z\in\mathcal{B}(\mathcal{H}_{1})^{\otimes m}\) on sites \(n+1,\dots,n+m\), so that \(Y\otimes Z\) acts on \(n+m\) sites. Introducing multi-indices \(k=(k_1,\dots,k_n)\), \(r=(r_1,\dots,r_m)\), and similarly \(\ell,s\), the matrix elements factorise as
\[
\langle k,r|Y\otimes Z|\ell,s\rangle
=
\langle k|Y|\ell\rangle\,\langle r|Z|s\rangle.
\]
By the definition of \(\widehat{\cdot}\),
\begin{align*}
\widehat{Y\otimes Z}(M)
&=
\sum_{k,r}\sum_{\ell,s}
\langle k,r|Y\otimes Z|\ell,s\rangle\,
A_{r_m}^{\dagger}\cdots A_{r_1}^{\dagger}A_{k_n}^{\dagger}\cdots A_{k_1}^{\dagger}\,
M\,
A_{\ell_1}\cdots A_{\ell_n}A_{s_1}\cdots A_{s_m} \\
&=
\sum_{k,\ell}\langle k|Y|\ell\rangle
\sum_{r,s}\langle r|Z|s\rangle\,
A_{r_m}^{\dagger}\cdots A_{r_1}^{\dagger}
\bigl(A_{k_n}^{\dagger}\cdots A_{k_1}^{\dagger}MA_{\ell_1}\cdots A_{\ell_n}\bigr)
A_{s_1}\cdots A_{s_m}
\end{align*}
Regrouping the summations in the above expression, one gets
\begin{align*}
\widehat{Y\otimes Z}(M)&=
\sum_{r,s}\langle r|Z|s\rangle\,
A_{r_m}^{\dagger}\cdots A_{r_1}^{\dagger}
\Bigl(\sum_{k,\ell}\langle k|Y|\ell\rangle A_{k_n}^{\dagger}\cdots A_{k_1}^{\dagger}MA_{\ell_1}\cdots A_{\ell_n}\Bigr)
A_{s_1}\cdots A_{s_m}\\
&= (\widehat{Z}\circ\widehat{Y})(M)
\end{align*}
 This shows \eqref{eq:hatY_composition} and finishes the proof.
\end{proof}

\section{Appendix: Hidden Quantum Markov Models}\label{app:A}

Given a Hilbert space $\mathcal{H}$, we denote by $\mathcal{B}(\mathcal{H})$ the
C$^{*}$-algebra of all linear operators on $\mathcal{H}$.
If $\mathcal{A}$ and $\mathcal{B}$ are finite-dimensional C$^{*}$-algebras, a linear map
$\mathcal{E}:\mathcal{A}\to\mathcal{B}$ is called \emph{positive} if
$\mathcal{E}(X)\succeq 0$ for every positive element $X\succeq 0$ in $\mathcal{A}$.
The map $\mathcal{E}$ is \emph{completely positive} (CP) when, for each $n\in\mathbb{N}$,
its canonical amplification
\[
\mathcal{E}^{(n)}:M_{n}(\mathcal{A})\longrightarrow M_{n}(\mathcal{B}),
\qquad
\mathcal{E}^{(n)}([X_{ij}]) := [\mathcal{E}(X_{ij})],
\]
is positive. A CP map $\mathcal{E}:\mathcal{A}\otimes\mathcal{B}\to\mathcal{A}$ is called a
\emph{transition expectation} if it is unital, that is,
\(
\mathcal{E}(\mathbb{I}_{\mathcal{A}}\otimes\mathbb{I}_{\mathcal{B}})=\mathbb{I}_{\mathcal{A}}.
\)
In the operator-algebraic formulation of quantum Markov chains, such maps
provide the noncommutative analogue of conditional expectations, encoding
one-step updates that “integrate out’’ the second tensor component while
preserving the order structure and the unit of $\mathcal{A}$ \cite{BR}.

In the concrete situation $\mathcal{A}=\mathcal{B}(\mathcal{H})$ and
$\mathcal{B}=\mathcal{B}(\mathcal{K})$, any transition expectation
$\mathcal{E}:\mathcal{B}(\mathcal{H})\otimes\mathcal{B}(\mathcal{K})\to\mathcal{B}(\mathcal{H})$
admits a Kraus representation
\[
\mathcal{E}(X) = \sum_{i=1}^{r} V_{i}^{*} X V_{i},
\qquad
X\in\mathcal{B}(\mathcal{H}\otimes\mathcal{K}),
\]
where $V_{i}:\mathcal{H}\to\mathcal{H}\otimes\mathcal{K}$ are linear operators satisfying
the normalization condition $\sum_{i} V_{i}^{*}V_{i} = \mathbb{I}_{\mathcal{H}}$.
The operators $V_{i}$ are referred to as the Kraus operators of $\mathcal{E}$.

The convex set of (normalized) density operators on $\mathcal{H}$ is
\[
\mathfrak{S}(\mathcal{H})
:=
\bigl\{\rho\in\mathcal{B}(\mathcal{H})
\;\big|\;
\rho\succeq 0,\ \operatorname{Tr}(\rho)=1\bigr\},
\]
whose elements represent quantum states. The Hilbert–Schmidt pairing
\[
\langle\rho,X\rangle
:=
\operatorname{Tr}(\rho X),
\qquad
\rho\in\mathfrak{S}(\mathcal{H}),\ X\in\mathcal{B}(\mathcal{H}),
\]
induces a unique adjoint for any CP, unital map
$\mathcal{E}:\mathcal{B}(\mathcal{H}\otimes\mathcal{K})\to\mathcal{B}(\mathcal{H})$.
More precisely, there exists a dual map
$\mathcal{E}_{*}:\mathfrak{S}(\mathcal{H})\to\mathfrak{S}(\mathcal{H}\otimes\mathcal{K})$
such that
\begin{equation}\label{eq_HS}
\operatorname{Tr}\!\bigl(\mathcal{E}_{*}(\rho)\,X\bigr)
=
\operatorname{Tr}\!\bigl(\rho\,\mathcal{E}(X)\bigr),
\qquad
\rho\in\mathfrak{S}(\mathcal{H}),\ X\in\mathcal{B}(\mathcal{H}\otimes\mathcal{K}).
\end{equation}
If $\mathcal{E}$ is written in Kraus form as above, then
\[
\mathcal{E}_{*}(\rho)
=
\sum_{i} V_{i}\rho V_{i}^{*},
\]
and $\mathcal{E}_{*}$ is a completely positive, trace-preserving map on density
operators, that is, a quantum channel.

In the Schr\"odinger picture, a channel
$\Phi:\mathfrak{S}(\mathcal{H})\to\mathfrak{S}(\mathcal{H}')$ encodes a causal influence
from the input system $\mathcal{H}$ to the output system $\mathcal{H}'$: manipulating
the input state can modify the statistics of all subsequent measurements on the
output, but not conversely.
The corresponding Heisenberg-picture map propagates observables backwards
along this causal arrow and realises the local dynamics of a quantum Markov chain.

We now fix the setting for hidden quantum Markov models.
Let $\mathcal{H}$ be an $N$-dimensional Hilbert space describing the hidden, or
internal, degrees of freedom, with a chosen orthonormal basis
$\{|j\rangle : j\in I_{H}\}$, where $I_{H}=\{1,\dots,N\}$.
We denote the associated hidden observable algebra by
\(\mathcal{B}_{H} := \mathcal{B}(\mathcal{H}).\)
Similarly, let $\mathcal{K}$ be an $M$-dimensional Hilbert space modelling the
output, or observation, space, with orthonormal basis
$\{|e_{k}\rangle : k\in I_{O}\}$, $I_{O}=\{1,\dots,M\}$, and observable algebra
\(\mathcal{B}_{O} := \mathcal{B}(\mathcal{K}).\)
Informally, elements of $\mathcal{B}_{H}$ play the role of latent quantum
causes, while elements of $\mathcal{B}_{O}$ represent observable effects.

Time is discrete and indexed by $\mathbb{N}=\{0,1,2,\dots\}$.
At each time $n\in\mathbb{N}$ we consider copies
$\mathcal{H}_{n}\simeq\mathcal{H}$ and $\mathcal{K}_{n}\simeq\mathcal{K}$ with
local observable algebras
\[
\mathcal{B}_{H;n} := \mathcal{B}(\mathcal{H}_{n})\simeq\mathcal{B}_{H},
\qquad
\mathcal{B}_{O;n} := \mathcal{B}(\mathcal{K}_{n})\simeq\mathcal{B}_{O}.
\]
The joint hidden–output algebra at time $n$ is the spatial tensor product
\[
\mathcal{B}_{H,O;n} := \mathcal{B}_{H;n}\otimes\mathcal{B}_{O;n},
\]
which collects all observables acting on the hidden subsystem and its
associated output at that instant.

A Hidden Quantum Markov Model (HQMM) provides a generative description of quantum stochastic processes through a bipartite structure separating hidden dynamics from observable outputs. Formally, an HQMM is defined as a state on the quasi-local algebra $\mathcal{A}_{H,O} = (\mathcal{B}(\mathcal{H}) \otimes \mathcal{B}(\mathcal{K}))^{\otimes \mathbb{N}}$, representing an infinite bipartite quantum system where $\mathcal{A}_{H} = \mathcal{B}(\mathcal{H})^{\otimes \mathbb{N}}$ describes the hidden system and $\mathcal{A}_{O} = \mathcal{B}(\mathcal{K})^{\otimes \mathbb{N}}$ represents the observable output system. The model is specified by a \emph{generative triplet} $\Xi = (\phi_0, \mathcal{E}_H, \mathcal{E}_{O,H})$, where:

\begin{itemize}
    \item $\phi_0: \mathcal{B}(\mathcal{H})\to \mathbb{C}$ is an initial state on the hidden algebra;
    \item $\mathcal{E}_{H}: \mathcal{B}(\mathcal{H}) \otimes \mathcal{B}(\mathcal{H}) \to \mathcal{B}(\mathcal{H})$ is a CPIP map governing hidden system dynamics;
    \item $\mathcal{E}_{O,H}: \mathcal{B}(\mathcal{H}) \otimes \mathcal{B}(\mathcal{K})\to \mathcal{B}(\mathcal{H})$ is a CPIP map encoding observable emission.
\end{itemize}

Two distinct causal structures emerge from different compositional orders of these maps, each with profound implications for quantum information transmission.

The coupling between hidden and observable degrees of freedom is described by
emission expectations
\[
\mathcal{E}_{H,O;n}:\mathcal{B}_{H;n}\otimes\mathcal{B}_{O;n}\to\mathcal{B}_{H;n},
\qquad n\ge 0.
\]
Their duals
\[
(\mathcal{E}_{H,O;n})_{*}:\mathfrak{S}(\mathcal{H}_{n})\to\mathfrak{S}(\mathcal{H}_{n}\otimes\mathcal{K}_{n})
\]
are quantum operations that generate an output system $\mathcal{K}_{n}$ conditionally
on the hidden state at time $n$, thereby implementing the causal arrow
$H_{n}\to O_{n}$.
If one traces out the hidden algebra, one obtains an induced process on
$\mathcal{B}_{O;\mathbb{N}}$ whose temporal correlations need not satisfy classical
Markov conditions; in process-tensor language, this apparent non-Markovianity
reflects genuine quantum memory propagated within $\mathcal{B}_{H;\mathbb{N}}$.

Given local elements $a_{n}\in\mathcal{B}_{H;n}$, $a_{n+1}\in\mathcal{B}_{H;n+1}$, and
$b_{n}\in\mathcal{B}_{O;n}$, we consider the one-step block maps
\begin{equation}\label{eq_F}
\mathcal{F}_{a_{n},b_{n}}^{(n)}(a_{n+1})
:=
\mathcal{E}_{H;n}\bigl(\mathcal{E}_{H,O;n}(a_{n}\otimes b_{n})\otimes a_{n+1}\bigr),
\end{equation}
and
\begin{equation}\label{eq_G}
\mathcal{G}_{a_{n},b_{n}}^{(n)}(a_{n+1})
:=
\mathcal{E}_{H,O;n}\bigl(\mathcal{E}_{H;n}(a_{n}\otimes a_{n+1})\otimes b_{n}\bigr),
\end{equation}
which encode two distinct causal orderings for the local dynamics.

\begin{definition}\label{CHQMMs}[Causal hidden quantum Markov model]
A \emph{causal hidden quantum Markov model} (causal HQMM) is a quadruple
\[
\Xi_{\mathrm{caus}}
=
\bigl(\phi_{H,0},(\mathcal{E}_{H;n})_{n\ge 0},(\mathcal{E}_{H,O;n})_{n\ge 0},(\mathcal{G}^{(n)})_{n\ge 0}\bigr),
\]
with the same initial hidden state \(\phi_{H,0}\) and the same families of hidden transition and emission expectations \((\mathcal{E}_{H;n})_{n\ge 0}\), \((\mathcal{E}_{H,O;n})_{n\ge 0}\) as above, but a different one-step composition rule encoded in the block maps \(\mathcal{G}^{(n)}\). For each \(n\), the block map
\[
\mathcal{G}^{(n)}:\mathcal{B}_{H;n}\otimes\mathcal{B}_{O;n}\otimes\mathcal{B}_{H;n+1}\to\mathcal{B}_{H;n}
\]
is defined, for \(a_n\in\mathcal{B}_{H;n}\), \(b_n\in\mathcal{B}_{O;n}\), \(X\in\mathcal{B}_{H;n+1}\), by
\[
\mathcal{G}^{(n)}_{a_n,b_n}(X)
:=
\mathcal{E}_{H,O;n}\bigl(\mathcal{E}_{H;n}(a_n\otimes X)\otimes b_n\bigr).
\]
For a local tensor \(\bigotimes_{m=0}^n (a_m\otimes b_m)\in\mathcal{B}_{H,O;[0,n]}\), the finite-time joint expectation of the causal HQMM is then
\begin{equation}\label{eq:caus-state-no-sup}
\psi_{H,O}\!\left(\bigotimes_{m=0}^n (a_m\otimes b_m)\right)
:=
\phi_{H,0}\circ
\mathcal{G}^{(0)}_{a_0,b_0}\circ\mathcal{G}^{(1)}_{a_1,b_1}\circ\cdots\circ
\mathcal{G}^{(n)}_{a_n,b_n}(\mathbb{I}_{H;n+1}),
\end{equation}
which again extends uniquely to a state \(\psi_{H,O}\) on \(\mathcal{B}_{H,O;\mathbb{N}}\). The associated hidden and observable marginals are
\[
\psi_H
:=
\psi_{H,O}\!\restriction_{\mathcal{B}_{H;\mathbb{N}}},\qquad
\psi_O
:=
\psi_{H,O}\!\restriction_{\mathcal{B}_{O;\mathbb{N}}}.
\]
\end{definition}

In the \emph{conventional} HQMM architecture, the hidden system at time $n$
first produces the observation at time $n$ via the emission map $\mathcal{E}_{H,O;n}$,
and is only afterwards propagated to time $n+1$ by the hidden transition
$\mathcal{E}_{H;n}$.
This emission–then–transition prescription underlies the standard HQMM update
rule and naturally appears in Heisenberg-picture descriptions of finitely
correlated and matrix product states driven by a hidden quantum memory.

By contrast, in the \emph{causal} HQMM architecture introduced here, the local
time order is reversed: the hidden configuration is first updated from
$H_{n}$ to $H_{n+1}$ by $\mathcal{E}_{H;n}$, and only the updated hidden state
acts as the cause of the observation at time $n$ through $\mathcal{E}_{H,O;n}$.
At the level of one-step maps, this corresponds to using the causal block map
$\mathcal{G}_{a_{n},b_{n}}^{(n)}$ in \eqref{eq_G} instead of the conventional
composition \eqref{eq_F}.
Since $\mathcal{E}_{H;n}$ and $\mathcal{E}_{H,O;n}$ do not commute in general, the
conventional and causal prescriptions give rise to inequivalent hidden quantum
Markov processes, even though they are built from the same local data
$(\phi_{H,0},(\mathcal{E}_{H;n})_{n\ge 0},(\mathcal{E}_{H,O;n})_{n\ge 0})$.

\section*{CRediT authorship contribution statement}

Abdessatar Souissi: Conceptualization, Investigation, Resources, Methodology, Writing – original draft, Writing – review \& editing.

Amenallah Andolsi: Visualization, Project administration, Methodology, Writing – original draft, Writing – review \& editing.

\section*{Data availability}

No data was used for the research described in the article.

\section*{Conflict of Interest}
The authors declare that there is no conflict of interest regarding the publication of this paper.

\section*{Funding}
The authors received no funding for this work.

\end{document}